\documentclass[10pt,conference,letterpaper]{IEEEtran}

\usepackage{amsmath,amssymb,amsthm}
\usepackage{algorithm}
\usepackage{algorithmic}
\usepackage{booktabs}
\usepackage{graphicx}
\usepackage[hidelinks]{hyperref}

\newtheorem{theorem}{Theorem}
\newtheorem{corollary}{Corollary}

\hypersetup{
  pdfauthor={Andrew Bond, I. Khilji},
  pdftitle={Budgeted Admission from Timestamped Capacity Records in Shared Clusters},
  pdfsubject={Distributed, Parallel, and Cluster Computing},
  pdfkeywords={Shared Clusters, Capacity Records, Budgeted Admission},
  pdfcreator={LaTeX}
}

\begin{document}

\title{Budgeted Admission from Timestamped\\
Capacity Records in Shared Clusters}

\author{
    \IEEEauthorblockN{Andrew Bond}
    \IEEEauthorblockA{\textit{Department of Computer Engineering} \\
    \textit{San Jose State University}\\
    San Jose, CA, USA \\
    andrew.bond@sjsu.edu}
    \and
    \IEEEauthorblockN{Syed Izhan Khilji} 
    \IEEEauthorblockA{\textit{Distributed Systems Group} \\
    \textit{TU Wien}\\
    Vienna, Austria \\
    ikhilji@dsg.tuwien.ac.at}
}

\maketitle

\begin{abstract}
A cluster admission gate can receive a capacity record that no longer matches
current schedulability. We study one exogenous binary schedulability unit under
a bound on the joint probability of admission and unavailability. For a
symmetric Markov state, we derive the stationary admission frontier for a fixed
observation delay and apply finite likelihood allocation to a delay-tagged
delivered record. A one-sided exact binomial endpoint calibrates each fitted
policy, and exact enumeration confirms the per-policy confidence target in
every tested parameter cell. A separate deterministic guard charges every
unresolved admission against the realized allowance and bounds
unavailable-state admissions at every prefix under truthful outcome feedback.
An executable admission procedure composes the allocation and guard. The
controlled system binds source sampling, NATS delivery, randomized admission,
Kubernetes objects, and API-derived outcomes in one append-only record.
Independent audit finds no missing event, identity mismatch, or state mismatch
in the completed functional block. A separate CPU request-accounting check
preserves Job and Pod identity across an unschedulable-to-complete transition.
\end{abstract}

\section{Introduction}
\label{sec-introduction}

Shared clusters often receive discretionary Jobs from users who do not
control placement. A guest-side gate can suppress an offer before it reaches
Kubernetes, but the gate may rely on a capacity record sampled by a separate
source and delivered over a networked control path. Stale allocable-capacity
views also arise in provider-side cloud admission~\cite{kerveros}. We study
the guest-side decision, where sample time, delivery time, decision time, and
scheduler outcome are distinct events.

We represent the target as one exogenous binary schedulability unit. A
positive action offers one Job. Available-state admission counts offers made
while the unit is available, while unavailable-state admission consumes a
prescribed probability budget. Both quantities are joint probabilities over
the action and the current state. They are not conditional rates among
admitted Jobs. The action does not affect the modeled state.

Three claims require different evidence. A known stationary model defines a
population optimum within the stationary policy class. A policy fitted from
finite transitions requires a confidence statement for that fitted policy. A
realized run whose outcomes arrive late requires an accounting rule that
remains valid before every label has been received. None of these statements
implies the other two.

Collision-constrained opportunistic access studies the corresponding
likelihood allocation for Markov on-off state, aged observations, and
randomized decisions~\cite{zhao2008periodic,chen2008osa,li2011cognitive}. We
specialize that decision basis to one timestamped schedulability record,
combine it with finite-sample calibration and unresolved-outcome accounting,
and connect each admission to its source sample, NATS record, Kubernetes
objects, and API-derived outcome.

We make three contributions.

\begin{itemize}
\item We specify the event contract from a timestamped capacity sample to an
API-derived outcome. For the symmetric one-unit Markov model, we give a
closed-form stationary fixed-delay benchmark and the finite-cell ordering for
a memoryless policy that consumes one delay-tagged delivered record.
\item We calibrate each fitted fixed-delay policy with a one-sided exact
binomial endpoint~\cite{clopperpearson}. The resulting per-policy statement
holds with the prescribed confidence over calibration data under the stated
model. A separate prefix guard charges every unresolved admission as
unavailable and bounds every realized prefix for arbitrary binary state and
proposal sequences with truthful feedback. We give an executable procedure
that composes the finite-cell allocation and guard.
\item We implement the event contract in a controlled NATS and Kubernetes
path with timestamped records, Jobs, and API-derived Job and Pod outcomes.
The functional audit recomputes policy decisions, prefix accounting, object
identity, and outcome agreement from recorded events. A separate repeated
check verifies an identity-preserving Kubernetes transition under CPU request
accounting.
\end{itemize}

The next section defines the record and outcome contract shared by the
analytical model and the controlled integration study.

\section{Record Contract and System Scope}
\label{sec-context}

\subsection{Admission boundary}

The guest gate decides whether to offer one Job to the host scheduler. It does
not place the Job, reserve capacity, or preempt another tenant. The analytical
state represents one exogenous schedulability unit. The controlled integration
study realizes that state through one run-owned scheduling taint. This signal
provides a known binary state for evaluating the admission path. It does not
represent physical device occupancy or cluster throughput.

\subsection{Timestamped capacity record}

At decision epoch $t$, the policy consumes a record containing a source
identifier, the sampled bit $X_{t-A_t}$, its logical age $A_t$, a source
timestamp, a receive timestamp, and a decision timestamp. The logical age
counts state transitions and selects the conditional law used by the policy.
The wall-clock interval from source sampling to the decision is recorded as
source-to-decision elapsed time. It is analyzed separately from logical age.

NATS request reply and leaf connections transport messages across clients and
NATS systems~\cite{natsrequest,natsleaf}. These mechanisms carry a record but
do not establish when its source state was sampled. Transport elapsed time
therefore cannot substitute for record age. A Kubernetes
\texttt{resourceVersion} identifies the object version returned by a
read~\cite{k8sapi}. The analysis treats this identifier as opaque and never as
a clock.

\subsection{Outcome record}

Only a positive admission creates a Job. The controlled harness releases an
outcome label only after the required Job and Pod conditions hold and the
recorded object identifiers and decision nonce match their predeclared values.
Verified completion maps to an available-state label. Scheduling rejection
caused by the run-owned taint maps to an unavailable-state label. Any other
terminal state, timeout, API failure, or identifier mismatch invalidates the
experimental attempt. The stochastic model treats the capacity process as
exogenous. The prefix theorem requires every released binary label to be
truthful under this contract.

\section{Model and Exact Frontier}
\label{sec-model}

The stationary capacity process is indexed by $t\in\mathbb{Z}$. Its state is
$X_t\in\{0,1\}$, where one denotes an available unit. The action is
$M_t\in\{0,1\}$. We restrict the frontier below to stationary randomized
policies. We define their one-epoch joint probabilities
\begin{equation}
g=\Pr(M_t=1,X_t=1)
\quad\text{and}\quad
\rho=\Pr(M_t=1,X_t=0).
\label{eq-objectives}
\end{equation}
The controller maximizes available-state admission $g$ subject to
$\rho\le\epsilon$, where $0\le\epsilon\le1/2$. The quantity $\rho$ is not
the conditional probability of an unavailable state given admission.

The state follows the symmetric Markov kernel
$\Pr(X_{t+1}=1-i\mid X_t=i)=p$ for $i\in\{0,1\}$ and
$p\in[0,1/2]$. Its marginal is uniform. This condition selects the uniform
stationary mixture at $p=0$, where the chain is not irreducible. The action
does not affect the state.

At epoch $t$, a stationary randomized policy may use the state history
available through $t-D$ and independent randomization. The fixed observation
delay $D$ is a nonnegative integer. Let $Y_t=X_{t-D}$. The probability of
disagreement is
\begin{equation}
q_D=\Pr(X_t\ne Y_t)
=\frac{1-(1-2p)^D}{2}.
\label{eq-q}
\end{equation}
The Markov property gives conditional independence of $X_t$ and the older
history given $Y_t$. Any history-dependent policy induces an effective pair
$\theta_i=\Pr(M_t=1\mid Y_t=i)$ with the same $(g,\rho)$. Hence an optimum can
be chosen to depend only on $Y_t$.

\begin{theorem}[Fixed-delay stationary frontier]
\label{thm-frontier}
For $0<q_D<1/2$, the optimum over the stated stationary policy class first
assigns admission probability to the observation $Y_t=1$. Its action
probabilities are
\begin{equation}
(\theta_1,\theta_0)=
\begin{cases}
(2\epsilon/q_D,0), & 0\le\epsilon\le q_D/2,\\
(1,(2\epsilon-q_D)/(1-q_D)), & q_D/2<\epsilon\le1/2.
\end{cases}
\label{eq-policy}
\end{equation}
The resulting frontier is
\begin{equation}
g^\star(\epsilon,D)=
\begin{cases}
\epsilon(1-q_D)/q_D, & \epsilon\le q_D/2,\\
(1-q_D)/2+
\dfrac{q_D}{1-q_D}(\epsilon-q_D/2),
& \epsilon>q_D/2.
\end{cases}
\label{eq-frontier}
\end{equation}
\end{theorem}

\begin{proof}
Stationarity and symmetry give
\begin{align}
g&=\tfrac12[(1-q_D)\theta_1+q_D\theta_0],\\
\rho&=\tfrac12[q_D\theta_1+(1-q_D)\theta_0].
\end{align}
The available-state to unavailable-state mass ratio is $(1-q_D)/q_D$ after
$Y_t=1$ and its reciprocal after $Y_t=0$. The first ratio is larger. The
two-variable linear program increases $\theta_1$ until it reaches one, then
increases $\theta_0$. Direct substitution gives \eqref{eq-policy} and
\eqref{eq-frontier}.
\end{proof}

This frontier is the closed-form two-cell specialization of the established
collision-constrained likelihood allocation reviewed in
Section~\ref{sec-related}.

When $q_D=0$, the minimal-budget boundary optimum
$(\theta_1,\theta_0)=(1,0)$ obtains $g=1/2$ with $\rho=0$. When
$q_D=1/2$, the observation is uninformative and
$g^\star=\epsilon$. The formula in \eqref{eq-frontier} extends continuously
to this boundary.

\begin{theorem}[Finite observation likelihood ratio rule]
\label{thm-binary}
Let $O_t$ be any observation in a finite alphabet $\mathcal O$. Define
$v_o=\Pr(O_t=o,X_t=1)$ and $w_o=\Pr(O_t=o,X_t=0)$. Among policies based on
$O_t$, an optimum fully admits every cell with $w_o=0<v_o$. It then admits
positive-cost cells in nonincreasing order of $v_o/w_o$ until the budget is
exhausted or every cell with $v_o>0$ is fully admitted. When the ratios are
distinct, at most one positive-cost cell is partially admitted. Equal-ratio
boundary cells may share the randomized mass. A cell with $v_o=0$ need not be
admitted. If $w_o=v_o=0$, its action has no effect.
\end{theorem}

\begin{proof}
The policy solves
\begin{equation}
\max_{0\le\theta_o\le1}\sum_{o\in\mathcal O} v_o\theta_o
\quad\text{subject to}\quad
\sum_{o\in\mathcal O} w_o\theta_o\le\epsilon.
\label{eq-binary-lp}
\end{equation}
Consider two positive-cost cells with ratios $r_i>r_j$, where
$\theta_i<1$ and $\theta_j>0$. Moving a sufficiently small unavailable-state
mass from cell $j$ to cell $i$ changes available-state mass by that amount
times $r_i-r_j>0$. An optimum therefore has no such inversion. Zero-cost and
tied cells give the stated boundary cases. If every cell with positive
available-state mass is fully admitted before the constraint binds, the
remaining budget may stay unused.
\end{proof}

The ratio ordering in Theorem~\ref{thm-binary} is the finite-observation
admission specialization of the randomized likelihood ratio rule of Neyman
and Pearson~\cite{neymanpearson}. The Markov assumption supplies the joint
masses in closed form and makes one delayed bit sufficient. A finite-context
model could instead estimate these masses without assuming Markov dynamics.

\begin{corollary}[Delivered-record allocation]
\label{cor-record-delay}
Let the delivered-record delay $A_t$ have finite support
$\mathcal A\subset\mathbb Z_{\ge0}$, with
$\pi_a=\Pr(A_t=a)>0$ for $a\in\mathcal A$. Assume that $A_t$ is independent of
the capacity process. Consider memoryless policies whose only state input is
$O_t=(A_t,X_{t-A_t})$. For $p\in(0,1/2)$, an optimum in this class may be
chosen to admit observed-available cells from smaller delay to larger delay.
If budget remains after these cells, it admits positive-delay
observed-unavailable cells from larger delay to smaller delay. At most one
positive-cost boundary cell needs randomization. Admission in a zero-delay
observed-unavailable cell never increases available-state admission and may be
omitted.
\end{corollary}

\begin{proof}
For $q_a=[1-(1-2p)^a]/2$, the available-state and unavailable-state masses of
an observed-available cell are
$\pi_a(1-q_a)/2$ and $\pi_a q_a/2$. Their ratio is
$(1-q_a)/q_a$. The corresponding ratio after an observed-unavailable sample is
$q_a/(1-q_a)$. For positive delay, the first ratio exceeds one and decreases
with delay. The second is below one and increases with delay.
Theorem~\ref{thm-binary} gives the stated order. At zero delay, the
observed-available cell has zero unavailable-state mass and is admitted first.
The observed-unavailable cell has zero available-state mass.
\end{proof}

Fixed $D$ is the special case in which every delivered record has age $D$. A
timestamp lets the controller select the appropriate delivered-record cell
for each decision. This result does not optimize a controller that retains
past deliveries. A past delivery can contain a newer state sample than the
current record.

The delay-hidden controller sees the delivered bit but not its record delay. Under the
independence assumption, its disagreement probability is
\begin{equation}
\bar q=\sum_{a\in\mathcal A}\pi_a q_a.
\label{eq-mixture-q}
\end{equation}
Substituting $\bar q$ into Theorem~\ref{thm-frontier} gives the optimum among
memoryless policies based only on that bit. This controller is the delay-hidden
comparison used in Section~\ref{sec-evaluation}.

At $p=0$, every observed state is exact, so admitting observed-available cells
and rejecting observed-unavailable cells is optimal at every delay. At
$p=1/2$, the zero-delay sample remains exact and all positive-delay cells have
equal ratio one.
The monotone delay ordering need not hold when record delay depends on the
capacity process.

For $0<q_D<1/2$, a record-independent randomized policy admits with
probability $2\epsilon$ and obtains $g_{\rm ind}=\epsilon$. The largest
vertical difference from \eqref{eq-frontier} occurs at
$\epsilon=q_D/2$ and equals
\begin{equation}
\max_{0\le\epsilon\le1/2}[g^\star(\epsilon,D)-\epsilon]
=\frac{1-2q_D}{2}
=\frac{(1-2p)^D}{2}.
\label{eq-maxdifference}
\end{equation}
For $0<q_D<1/2$, the difference increases with $\epsilon$ on the first branch
and decreases on the second branch, so its maximum occurs at the branch point.
At a specified budget, the increase is
$g^\star(\epsilon,D)-\epsilon$. Equation~\eqref{eq-maxdifference} instead
maximizes that increase over all admissible budgets.

\begin{figure*}[t]
\centering
\includegraphics[width=0.98\textwidth]{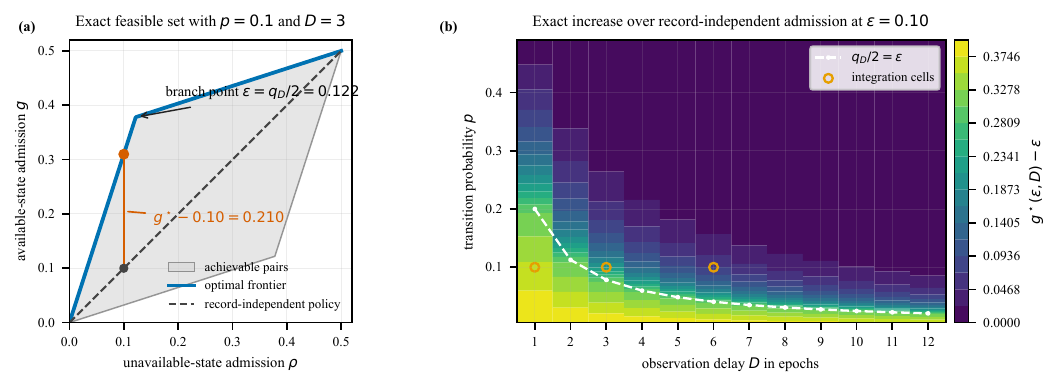}
\caption{Exact quantities from the symmetric Markov model. Panel a shows the
feasible pairs and complete upper frontier for $p=0.10$ and $D=3$. The orange
segment gives the increase over record-independent admission at
$\epsilon=0.10$. Panel b evaluates that increase over integer observation
delays and transition probabilities. The white curve joins the exact branch
points at those integers. The outlined cells identify the fixed delays used in
the controlled integration study. Both panels are analytical evaluations of
the model.}
\label{fig-frontier}
\end{figure*}

\begin{corollary}
For fixed $p\in(0,1/2)$, the delayed-state frontier converges to the
record-independent frontier as $D$ increases.
\end{corollary}

The corollary follows because $q_D$ converges to $1/2$ and both branches of
\eqref{eq-frontier} converge to $\epsilon$. A non-Markov process can remain
predictable from longer histories even when one delayed bit contains little
information about the current state.

\section{Calibration and Prefix Protection}
\label{sec-policy}

\subsection{Exact finite-sample calibration}

The transition indicators
$I_i=\mathbf{1}\{X_i\ne X_{i-1}\}$ are independent Bernoulli variables under
the symmetric chain. Let $K=\sum_{i=1}^{n}I_i$. Then
$K\sim{\rm Binomial}(n,p)$ for an integer $n\ge1$. The direct controller uses
$\widehat p=\min\{1/2,K/n\}$ in \eqref{eq-q} and
\eqref{eq-policy}. An underestimate can make the fitted policy infeasible.

For $0<\delta<1$, let $U(K)$ be the one-sided Clopper Pearson endpoint
~\cite{clopperpearson}. Its projection to the model interval is
\begin{equation}
p^+(K)=
\begin{cases}
\min\!\left\{\frac12,
B^{-1}_{1-\delta}(K+1,n-K)\right\}, & K<n,\\
\frac12, & K=n,
\end{cases}
\label{eq-cp}
\end{equation}
where $B^{-1}_{u}(a,b)$ is the $u$ quantile of a Beta distribution. Let
$q_D^+=[1-(1-2p^+)^D]/2$. The calibrated controller substitutes $q_D^+$
into \eqref{eq-policy} and uses the exact-record boundary policy when
$q_D^+=0$.

\begin{theorem}[Population budget after calibration]
\label{thm-calibration}
Fix an integer $n\ge1$, an integer $D\ge0$,
$0\le\epsilon\le1/2$, and $0<\delta<1$. For every
$p\in[0,1/2]$, form $K$ from $n$ chain transitions and apply the calibrated
policy to an independent stationary evaluation chain with the same $p$. Then
\begin{equation}
\Pr_{p,K}\!\left[
\Pr_p(M_t=1,X_t=0\mid K)\le\epsilon
\right]\ge1-\delta.
\label{eq-population}
\end{equation}
\end{theorem}

\begin{proof}
The unprojected endpoint has $\Pr_K(p\le U)\ge1-\delta$. Since $p\le1/2$,
projection preserves this event. Equation~\eqref{eq-q} is nondecreasing in
$p$, so coverage gives $q_D\le q_D^+$. If $q_D^+=0$, coverage implies
$q_D=0$ and the exact-record boundary policy has $\rho=0$. Otherwise,
unavailable-state mass on the first policy branch equals
$\epsilon q_D/q_D^+$. On the second branch it is affine and nondecreasing in
$q_D$, and equals $\epsilon$ at $q_D=q_D^+$. The boundary policy at
$q_D^+=1/2$ uses the continuous extension of \eqref{eq-policy}. Every case is
feasible on the covered event.
\end{proof}

Theorem~\ref{thm-calibration} covers the population joint probability of one
fitted stationary policy. Evaluation uses a fresh stationary run or a
disjoint future segment with transition increments that do not overlap the
$n$ calibration increments. The result neither bounds a finite empirical
fraction nor gives simultaneous coverage over a parameter sweep.

\subsection{Deterministic prefix guard}

We next bound realized unavailable-state admissions at every finite prefix.
Let $t\ge1$. Define
$C_t=1-X_t$, $H_0=0$, and
\begin{equation}
H_t=\sum_{s=1}^{t}M_sC_s
\label{eq-harm-count}
\end{equation}
as the realized unavailable-state admission count. Let $B_t$ be a binary base
proposal.
Before action $t$, let $R_t$ and $U_t$ partition the prior admission indices
$\{s<t\mid M_s=1\}$. The set $R_t$ contains admissions with truthful resolved
outcomes. The set $U_t$ contains admissions whose outcomes remain unresolved.
The controller-visible worst-case counter is
\begin{equation}
W_t=\sum_{s\in R_t}C_s+\lvert U_t\rvert.
\label{eq-worst-counter}
\end{equation}
Every unresolved admission is charged as unavailable. Let
$L_t=\lfloor\epsilon t\rfloor$. When $B_t=1$ and $W_t+1\le L_t$, the guard
sets $M_t=1$. In every other case it sets $M_t=0$. Truthful outcomes may
arrive late, out of order, or not at all. A missing outcome remains charged.

\begin{theorem}[Pathwise prefix budget]
\label{thm-shield}
For every binary state sequence, every binary proposal sequence, and every
feedback schedule whose reported outcomes are truthful, the guard satisfies
\begin{equation}
H_t\le\lfloor\epsilon t\rfloor
\qquad\text{for every }t.
\label{eq-prefix}
\end{equation}
\end{theorem}

\begin{proof}
The empty history has $W_1=0\le L_1$. Suppose the pre-action invariant
$W_t\le L_t$ holds. The unavailable-state count through $t$ is at most
$W_t+M_t$. If $M_t=1$, the guard gives $W_t+M_t\le L_t$. If $M_t=0$, the
pre-action invariant gives the same inequality. Charging an admitted action
as unresolved makes the post-action worst-case counter $W_t+M_t$.
Resolving an old outcome replaces
one charged admission by its truthful zero or one label and cannot increase
the counter. Hence $W_{t+1}\le W_t+M_t\le L_t\le L_{t+1}$. Induction proves
both the counter invariant and \eqref{eq-prefix}.
\end{proof}

\begin{theorem}[Pointwise maximal prefix guard]
\label{thm-maximal-shield}
Fix a pre-action history at epoch $t$ and a proposal $B_t=1$. Among binary
pass-through wrappers evaluated at that same history, passing the proposal is
safe for prefix $t$ under every binary completion of $U_t$ and the current
outcome if and only if $W_t+1\le L_t$. The guard is therefore pointwise
maximally permissive for the stated information and completion class.
\end{theorem}

\begin{proof}
If $W_t+1\le L_t$, the current prefix count is at most $W_t+1\le L_t$
for every completion. Theorem~\ref{thm-shield} covers later prefixes under the
same guard. If the inequality fails, assign the unavailable label to every
unresolved admission and to the current action. The current count is then
$W_t+1>L_t$. Any wrapper that admits violates the bound on this completion.
\end{proof}

\begin{algorithm}[t]
\caption{Timestamped-record admission with prefix protection}
\label{alg-admission}
\footnotesize
\begin{algorithmic}
\STATE Form finite cells $o$ with joint masses $(v_o,w_o)$
\STATE Set $\theta_o\gets0$ for every cell and set $r\gets\epsilon$
\STATE Set $\theta_o\gets1$ for every cell with $w_o=0<v_o$
\STATE Order cells with $v_o>0$ and $w_o>0$ by decreasing $v_o/w_o$
\FOR{each ordered positive-cost cell $o$}
\STATE $\theta_o\gets\min\{1,r/w_o\}$
\STATE $r\gets r-w_o\theta_o$
\ENDFOR
\STATE Set confirmed count $C\gets0$ and unresolved set $U\gets\varnothing$
\FOR{decision epochs $t=1,2,\ldots$}
\FOR{each verified label $(s,c_s)$ released before decision $t$}
\STATE $U\gets U\setminus\{s\}$
\STATE $C\gets C+c_s$
\ENDFOR
\STATE Validate and read $o_t=(A_t,X_{t-A_t})$
\STATE Draw $Z_t$ uniformly and set $B_t\gets\mathbf 1\{Z_t<\theta_{o_t}\}$
\STATE $W_t\gets C+\lvert U\rvert$
\STATE $M_t\gets\mathbf 1\{B_t=1\text{ and }W_t+1\le
\lfloor\epsilon t\rfloor\}$
\IF{$M_t=1$}
\STATE Offer one Job and add its decision identifier to $U$
\ENDIF
\STATE Append the record, draw, decision, counter, and released labels
\ENDFOR
\end{algorithmic}
\end{algorithm}

Algorithm~\ref{alg-admission} composes the finite-cell allocation with the
prefix guard. The first phase implements Theorem~\ref{thm-binary}. A calibrated
fixed-delay policy obtains its cell masses from $p^+$ in
\eqref{eq-cp}. The second phase does not trust those fitted masses. Its prefix
bound follows from Theorem~\ref{thm-shield} for any proposal sequence when the
released binary labels are truthful.

The supplementary model represents binary traces with FIFO outcome resolution
and zero charge at rejected epochs. The Go library uses the same integer
allowance but resolves feedback by action identifier. It serializes counter
updates, rejects fixed-width overflow, and records every release for
independent recomputation. The pathwise theorem applies when the counter
inequalities hold and every released label satisfies the outcome contract in
Section~\ref{sec-context}.

\section{Instrumented Admission Path}
\label{sec-implementation}

\subsection{Controlled integration harness}

The functional check uses a controlled harness. It exercises NATS delivery,
the timestamped-record policy, guard decisions, Kubernetes Job creation,
scheduling, completion, and outcome accounting. The completed guarded
treatments have zero allowance and create no Job. Kubernetes outcomes in this
block therefore come from the unguarded treatments. The harness calls the
Kubernetes client directly, so its measurements concern this controlled path
rather than the general request bridge.

\begin{figure*}[t]
\centering
\includegraphics[width=0.99\textwidth]{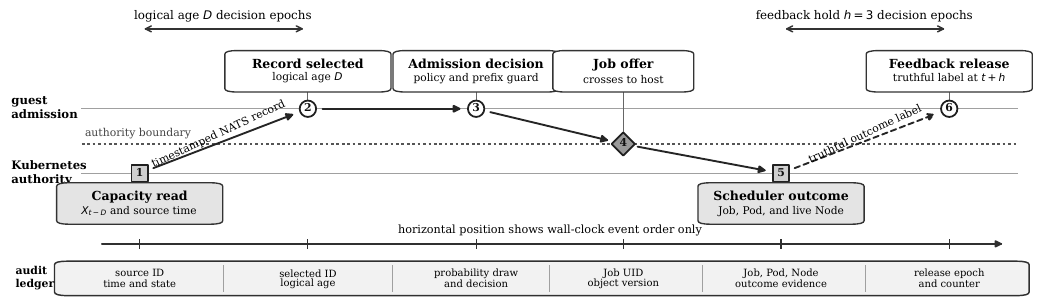}
\caption{Controlled event contract. The upper and lower lanes separate guest
admission from host scheduler authority. Numbered events follow wall-clock
order. The brackets define logical record age and the injected
controller-visible feedback hold in decision epochs. They do not represent
transport latency. Ledger fields bind the source record, decision, Kubernetes
objects, outcome label, and counter update.}
\label{fig-system-path}
\end{figure*}

One runner controls a dedicated single-node development cluster. Namespaced
objects are isolated in an experiment namespace. The Node taint is
cluster-scoped. At each logical epoch the runner applies or removes one
run-specific value under a dedicated \texttt{NoSchedule} taint key. After a
100 ms settling interval, the runner reads the Node through the API, records
the returned taint state, UID, and opaque \texttt{resourceVersion}, then
assigns a local sample timestamp immediately after that
read~\cite{k8staints,k8sapi,k8sobjects}. The analysis never parses
\texttt{resourceVersion} or treats it as a clock.

After each API read, the runner stores the current sample in a local record
buffer. For a treatment with logical age $A_t$, it selects the entry whose
trace index is $t-A_t$ and publishes that selected entry to the policy subject.
NATS therefore carries an intentionally old timestamped record. The age $A_t$
is injected by trace-index selection and is neither produced nor estimated
from transport latency. Distinct NATS connections publish and consume this
policy-facing record. The delivery contains the source identifier, sampled
state, logical record age in epochs, source time, publish time, and receive
time. The policy records its admission probability, common random draw,
proposal, and decision time. Logical record age and source-to-decision elapsed
time remain separate quantities.

A positive action creates a Kubernetes Job. Its Pod uses a hostname
\texttt{nodeSelector} for the selected node and has no toleration for the
experiment taint. The taint therefore prevents scheduling in the controlled
unavailable state~\cite{k8staints}. The runner holds that state until it records
Job completion or a taint-caused scheduling rejection, then deletes the Job
before the next state transition.

The outcome classifier reads the Job and Pod API objects together with the
predeclared taint identity and a decision-specific nonce. Both accepted
outcomes require the exact Job annotation, Job UID, Pod UID, and Pod
owner reference. A verified completion also requires a true
\texttt{Complete} Job condition, a true \texttt{PodScheduled} condition,
container exit code zero, and a Kubernetes-populated termination
message~\cite{k8sjobs,k8spods,k8stermination}. The harness accepts that message
only when it equals the predeclared nonce. A controlled unavailability outcome
requires \texttt{PodScheduled} to be false with reason
\texttt{Unschedulable}. The classifier then reads the Node again and requires
the same Node UID, a true Ready condition, a schedulable Node, the run-owned
\texttt{NoSchedule} taint, the expected hostname selector, and no matching Pod
toleration. Kubernetes condition text is retained for diagnosis but never
determines the outcome. Any API failure, timeout, terminal Job failure,
nonzero exit, identity mismatch, or other unclassified outcome aborts the
classifier and invalidates the run.

The classifier obtains each admitted action outcome before the next capacity
transition. For a guarded admission, the controller releases that API-derived
label after a prescribed three-epoch hold. Labels due at the same epoch are
released newest first, including any labels released during the final flush.
The completed four-epoch treatments did not admit a guarded Job and therefore
did not exercise this delayed guarded-feedback transition. The pathwise result
covers every truthful release order. Infrastructure failures instead
invalidate the run under the classifier rules above.

All treatments in the functional block use the same state sequence, record
delays, offered requests, and random draws. The block is used for path
coverage rather than policy effect estimation. The append-only writer flushes
and synchronizes each JSON record and refuses to overwrite an existing file.
The analyzer independently recomputes the policy probability, proposal, guard
state, deterministic identifiers, event ordering, and prefix allowance. It
checks the recorded Kubernetes UIDs and computes $M_tX_t$ and
$M_t(1-X_t)$ from the current state. Job completion and taint rejection
separately verify that the controlled state produced the intended Kubernetes
behavior.

\subsection{Repository request gate}

The repository also contains a serialized gate for a general NATS request
bridge. It rejects missing or denied observations before Job creation and
publishes the same status to the canonical subject and reply inbox. The
controlled experiment calls the Kubernetes client directly and does not
measure this bridge.

\section{Evaluation}
\label{sec-evaluation}

\subsection{Evidence roles}

The evaluation separates calculation from execution evidence. The analytical
figures evaluate the stated probability model. The calibration sweep sums the
complete finite-sample distribution. A controlled functional block checks the
record-to-outcome implementation path but is not used to estimate policy
effects. A separate repeated check tests a Kubernetes capacity transition
driven by CPU request accounting.

\subsection{Controlled functional path}

The functional block uses the controlled state process with $p=0.10$ and
$\epsilon=0.10$. It exercises record-independent, fixed-delay,
delay-hidden, delay-tagged, and calibrated policies with and without the
prefix guard. Fixed record ages are one, three, and six decision epochs.
Variable ages use the same three values. Each of the 16 treatments executes
four evaluation decisions on one shared state path. The resulting 64 decisions
are a functional coverage record rather than an independent-unit performance
study. We therefore report event integrity and local path timing but no policy
effect estimate. For every guarded treatment, the allowance
$\lfloor\epsilon t\rfloor$ is zero at all four decisions. The 32 guarded
decisions therefore produce no action or delayed outcome release. All positive
actions in this block come from unguarded treatments.

The independent audit found all 64 policy deliveries and all 64 decisions.
Twenty-six positive actions produced 26 classified Kubernetes outcomes.
Twelve Jobs completed in the controlled available state and 14 Pods were
unschedulable under the run-owned taint. There was no missing delivery,
outcome, or API result. The audit found no identifier, record-age,
state-to-outcome, or event-order mismatch. Recomputed guard decisions and
prefix allowances matched the record. Zero prefix excess follows from zero
guarded actions and is not a positive-action guard test. NATS
publish-to-receive time had median 1.593 ms and type-seven empirical 95th
percentile 2.554 ms. These local timings describe the exercised port forward
and are not multi-host latency measurements.

\subsection{Kubernetes request accounting check}

We also tested one capacity transition produced by Kubernetes CPU request
accounting~\cite{k8sresources}. A locally frozen protocol ran 32 sequential
cycles on one Ready untainted node. Existing Pods requested 960 mCPU from
18000 mCPU allocatable. A sleeping holder requested 17015 mCPU and left
25 mCPU, so a 50 mCPU target did not fit. Kubernetes reported
\texttt{PodScheduled=False} with reason \texttt{Unschedulable} and message
\texttt{Insufficient cpu} in every cycle.

After verified holder deletion, the same Job UID and the same Pod UID completed
in all 32 cycles. Every completion had exit code zero and the exact termination
nonce. The background request inventory returned to its initial identities and
request total. Holder deletion to observed completion had median 3.573 s and
range 2.967 to 3.625 s. This interval includes scheduling, container startup,
execution, API observation, and polling. The sleeping holder generated no
synthetic CPU work. The check concerns request accounting and object lifecycle,
not CPU utilization, throughput, or interference.

\begin{table*}[t]
\centering
\caption{Audited execution evidence. The functional block checks every record
and decision together with the unguarded Kubernetes outcome path. Its guarded
arms cover zero-allowance decisions only. The CPU request check repeats one
identity-preserving scheduler transition on the same node.}
\label{tab-execution}
\small
\begin{tabular}{@{}p{0.48\textwidth}p{0.44\textwidth}@{}}
\toprule
Check & Result \\
\midrule
\multicolumn{2}{@{}l}{\emph{Controlled functional path}} \\
Evaluation decisions and delivered records & 64 and 64 \\
Positive actions and classified outcomes & 26 and 26 \\
Completed and controlled-unschedulable outcomes & 12 and 14 \\
Guarded decisions, actions, and delayed releases & 32, 0, and 0 \\
Missing delivery, outcome, and API result & 0, 0, and 0 \\
Identifier, record-age, state, and timing mismatch & 0, 0, 0, and 0 \\
Maximum guarded prefix excess & 0 under zero allowance \\
NATS publish-to-receive time & 1.593 ms median and 2.554 ms empirical p95 \\
\midrule
\multicolumn{2}{@{}l}{\emph{CPU request-accounting lifecycle check}} \\
Allocatable, baseline, holder, residual, and target request &
18000, 960, 17015, 25, and 50 mCPU \\
\texttt{Insufficient cpu} observed & 32 of 32 cycles \\
Same Job UID and Pod UID completed & 32 of 32 cycles \\
Background request inventory and nonce verified & 32 of 32 cycles \\
Failed cycle & 0 \\
Holder deletion to observed completion &
3.573 s median and 2.967 to 3.625 s range \\
\bottomrule
\end{tabular}
\end{table*}

\subsection{Exact calibration sweep}

The functional path evaluates one calibrated operating point. We separately
compute the calibration distribution without Monte Carlo sampling. The sweep
uses
$p\in\{0.005,0.01,0.02,0.05,0.1,0.2,0.35,0.49\}$,
$D\in\{1,3,8,12\}$,
$\epsilon\in\{0.02,0.05,0.1,0.2\}$, and seven calibration sizes from 50 to
5,000. Each size has 128 parameter cells. We compare the direct estimate, a
matched one-sided Hoeffding endpoint, and \eqref{eq-cp} with $\delta=0.01$.

\begin{figure*}[t]
\centering
\includegraphics[width=0.98\textwidth]{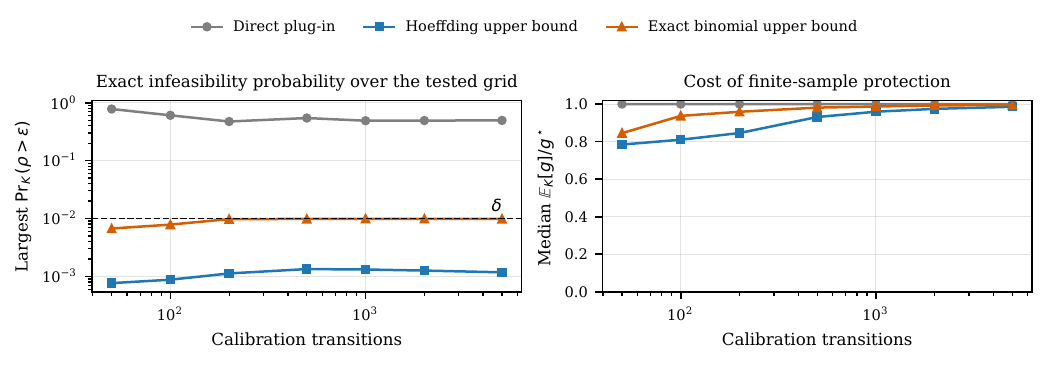}
\caption{Exact finite-sample calculation over the tested parameter grid. The
left panel reports the largest probability that a fitted population policy
violates its budget. The right panel reports the median ratio between expected
calibrated-policy admission and oracle admission. Expectation is taken over the
exact binomial calibration distribution within each cell. The direct rule can
exceed one because an infeasible fitted policy can exceed the feasible oracle.}
\label{fig-calibration}
\end{figure*}

At $n=2000$, the direct rule has a largest violation probability of 0.4916.
The Hoeffding endpoint reduces this value to 0.00126. The exact endpoint gives
0.00986 and remains below 0.01 at every tested calibration size. At $n=2000$,
the median ratio of expected calibrated admission to oracle admission is
0.9745 for the Hoeffding endpoint and 0.9927 for the exact endpoint. The
corresponding tenth percentiles are 0.5041 and 0.8265. The exact endpoint
retains more admission while meeting the per-policy confidence target
throughout the tested grid.

\subsection{Reproducibility checks}

The functional JSONL record is append only. The audit binds its digest and
independently recomputes record selection, policy probability, proposal,
counter transition, prefix allowance, object identity, event order, and
state-to-outcome agreement. The CPU request analyzer checks exact Job, Pod,
holder, Node, image, nonce, and request-inventory identities across every
cycle. It rejects a changed scheduler cause, object replacement, request
arithmetic error, background inventory change, timeout, or image drift.

Automated checks cover policy arithmetic, the exact endpoint, finite-cell
ordering, the serialized request gate, the outcome classifier, the prefix
counter, event-path auditing, request accounting, and deterministic figure
generation. No incomplete randomized-block artifact contributes to a reported
number or figure.

\section{Related Work}
\label{sec-related}

\textbf{Cloud admission and scheduling.}
Kerveros performs provider-side cloud admission from allocable VM counts and
documents stale cached views and optimistic concurrency~\cite{kerveros}. Borg
represents operator-side cluster scheduling~\cite{borg}, while KUBEDIRECT
reduces Kubernetes API exchange during burst provisioning~\cite{kubedirect}.
Our gate acts on the guest side before one Job offer. It neither packs
resources nor replaces the scheduler.

\textbf{Delayed state decisions.}
Delayed network state has long been studied in Markov control and wireless
scheduling~\cite{brooks1972,ying2011delayed}. Delayed queue admission also
predates this work~\cite{kuri1995delayed,altman1995monotonic}. The direct
mathematical antecedent is collision-constrained access with aged or erroneous
Markov observations~\cite{zhao2008periodic,chen2008osa,li2011cognitive}. Li et
al.\ derive memoryless access under tight collision constraints. Our
fixed-delay result is the one-unit closed form after converting conditional
collision probability to joint unavailable-state mass. The action does not
change the modeled state and the objective is not queue stability.

\textbf{Record age and calibration.}
Age of information is the time since generation of the newest delivered
update~\cite{aoi}. We use the generation time and logical age of the record
consumed by one decision. That record need not be the newest retained update.
Clopper and Pearson give exact binomial confidence limits
~\cite{clopperpearson}. Learn then Test gives finite-sample risk calibration
for general predictive algorithms through multiple testing
~\cite{angelopoulos2025ltt}. Our construction controls one fitted policy under
the symmetric-chain model and gives no simultaneous sweep guarantee.

\textbf{Delayed outcome protection.}
Runtime shields enforce modeled safety properties~\cite{bloem2015shield}.
Delay-resilient shields have been synthesized for delayed input observations
~\cite{canocordoba2023}. Our guard instead receives an outcome after the
action. It keeps every unresolved admitted outcome charged and bounds a
cumulative binary prefix. Its maximality statement is pointwise for one
pre-action counter under arbitrary binary completions.

\textbf{Kubernetes event evidence.}
Sieve and Kivi test Kubernetes controller observations and event
interleavings~\cite{sieve,kivi}. Our evidence does not verify Kubernetes. It
checks identity, order, and outcome agreement for one admission path from the
recorded client view.

\section{Limitations}
\label{sec-limitations}

The analytical frontier applies to one exogenous symmetric binary Markov unit
and stationary policies. It does not establish optimality for asymmetric,
multilevel, nonstationary, or action-dependent capacity. The calibration
theorem additionally requires independent Bernoulli transition indicators. If
the transition model is wrong, its population confidence statement need not
hold.

The variable-delay result covers memoryless policies that consume only the
current delivered pair. A controller that retains prior deliveries can have a
newer sample than the record currently delivered. The experiment samples
delivered-record delays independently. This distribution need not be the age
process seen by a receiver that always retains its freshest update. The
ordering also requires record delay to be independent of the capacity process.

The deterministic prefix guard requires no stochastic state model, but it can
reject available-state actions. An unresolved admission consumes allowance
until a truthful binary label is released. The pathwise result concerns an
abstract binary counter and truthful labels. It does not establish correctness
of the Go runtime, NATS, Kubernetes, or the correspondence between an API
condition and physical capacity.

The integration experiment controls schedulability through one
\texttt{NoSchedule} taint on one Docker Desktop node. It includes no incumbent
work, physical accelerator occupancy, queue competition, multi-unit placement,
or multi-tenant production load. A Pod held pending by the taint represents a
controlled unavailable-state action. It is not evidence of physical
interference or a tenant policy violation.

The functional block has four decisions per treatment on one shared path. It
tests event coverage and audit consistency but cannot estimate policy effects,
long-run admission, or retained admission under the prefix guard. Its guarded
arms have zero allowance at every tested decision. No guarded Job or delayed
guarded-feedback transition is exercised.

Capacity advances in logical decision epochs after the current action outcome
has been recorded. Treatments that create more Jobs therefore take more wall
time. The reported $g$ and $\rho$ values are per decision epoch rather than
cluster throughput. The NATS broker is reached through a local Kubernetes port
forward, so the study does not estimate multi-host transport latency.

The separate CPU request-accounting check uses a sleeping holder Pod on the
same node. It verifies one pending-to-complete lifecycle transition under
Kubernetes request accounting. It does not exercise NATS, an admission policy,
or the prefix guard. The holder generates no synthetic CPU work, so the check
does not measure utilization, interference, or physical load.

The audit verifies identifier equality and event order as observed by one
client. It treats \texttt{resourceVersion} as opaque and does not infer a
global Kubernetes event order from object versions, creation timestamps, or
condition transition timestamps~\cite{k8sapi}.

\section{Conclusion}
\label{sec-conclusion}

Timestamped capacity records make observation delay an explicit policy input.
For one symmetric binary Markov unit, the stationary fixed-delay frontier gives
an exact matched-budget benchmark. With finite transition data, a one-sided
binomial endpoint gives the stated population confidence for each fitted policy
under the model. A separate prefix guard bounds unavailable-state admissions
at every prefix while truthful outcomes remain unresolved. The admission
algorithm composes the finite-cell allocation with this guard. The completed
functional block audits every delivered record and decision. Its unguarded
positive actions connect the source sample, randomized decision, Job, Pod, and
verified outcome with no audited mismatch. Its guarded arms cover only
zero-allowance decisions. The repeated CPU request-accounting check preserves
Job and Pod identity across the tested scheduler transition.
Evaluation on uncontrolled multi-node capacity will require a capacity label
independent of the admission action and validation against a stated
operational objective.

\bibliographystyle{IEEEtran}
\bibliography{references}

@article{kuri1995delayed,
  author = {Kuri, Joy and Kumar, Anurag},
  title = {Optimal Control of Arrivals to Queues with Delayed Queue Length
           Information},
  journal = {IEEE Transactions on Automatic Control},
  year = {1995},
  month = aug,
  volume = {40},
  number = {8},
  pages = {1444--1450},
  publisher = {IEEE},
  issn = {0018-9286},
  doi = {10.1109/9.402238},
  url = {https://doi.org/10.1109/9.402238}
}

@article{altman1995monotonic,
  author = {Altman, Eitan and Stidham, Jr., Shaler},
  title = {Optimality of Monotonic Policies for Two-Action Markovian Decision
           Processes, with Applications to Control of Queues with Delayed
           Information},
  journal = {Queueing Systems},
  year = {1995},
  month = sep,
  volume = {21},
  number = {3--4},
  pages = {267--291},
  publisher = {Springer},
  issn = {0257-0130},
  doi = {10.1007/BF01149165},
  url = {https://doi.org/10.1007/BF01149165}
}

@inproceedings{kerveros,
  author = {Sultan Mahmud Sajal and Luke Marshall and Beibin Li and
            Shandan Zhou and Abhisek Pan and Konstantina Mellou and
            Deepak Narayanan and Timothy Zhu and David Dion and
            Thomas Moscibroda and Ishai Menache},
  title = {Kerveros: Efficient and Scalable Cloud Admission Control},
  booktitle = {17th USENIX Symposium on Operating Systems Design and
               Implementation},
  pages = {209--226},
  year = {2023},
  month = jul,
  address = {Boston, MA},
  isbn = {978-1-939133-34-2},
  publisher = {USENIX Association},
  url = {https://www.usenix.org/conference/osdi23/presentation/sajal}
}

@article{ying2011delayed,
  author = {Lei Ying and Sanjay Shakkottai},
  title = {On Throughput Optimality with Delayed Network-State Information},
  journal = {IEEE Transactions on Information Theory},
  volume = {57},
  number = {8},
  pages = {5116--5132},
  year = {2011},
  month = aug,
  doi = {10.1109/TIT.2011.2158879},
  url = {https://doi.org/10.1109/TIT.2011.2158879}
}

@article{canocordoba2023,
  author = {Cano C{\'o}rdoba, Filip and Alexander Palmisano and
            Martin Fr{\"a}nzle and Roderick Bloem and
            Bettina K{\"o}nighofer},
  title = {Safety Shielding under Delayed Observation},
  journal = {Proceedings of the International Conference on Automated
             Planning and Scheduling},
  volume = {33},
  number = {1},
  pages = {80--85},
  year = {2023},
  doi = {10.1609/icaps.v33i1.27181},
  url = {https://ojs.aaai.org/index.php/ICAPS/article/view/27181}
}

@article{angelopoulos2025ltt,
  author = {Anastasios N. Angelopoulos and Stephen Bates and
            Emmanuel J. Cand{\`e}s and Michael I. Jordan and Lihua Lei},
  title = {Learn then Test: Calibrating Predictive Algorithms to Achieve
           Risk Control},
  journal = {The Annals of Applied Statistics},
  volume = {19},
  number = {2},
  pages = {1641--1662},
  year = {2025},
  doi = {10.1214/24-AOAS1998},
  url = {https://doi.org/10.1214/24-AOAS1998}
}

@article{brooks1972,
  author = {D. M. Brooks and C. T. Leondes},
  title = {Technical Note---{Markov} Decision Processes with State-Information
           Lag},
  journal = {Operations Research},
  volume = {20},
  number = {4},
  pages = {904--907},
  doi = {10.1287/opre.20.4.904},
  year = {1972}
}

@inproceedings{borg,
  author = {Abhishek Verma and Luis Pedrosa and Madhukar R. Korupolu and
            David Oppenheimer and Eric Tune and John Wilkes},
  title = {Large-Scale Cluster Management at {Google} with {Borg}},
  booktitle = {Proceedings of the Tenth European Conference on Computer
  Systems},
  pages = {1--17},
  doi = {10.1145/2741948.2741964},
  year = {2015}
}

@article{aoi,
  author = {Roy D. Yates and Yin Sun and Brown, III, D. Richard and Sanjit K. Kaul
            and Eytan Modiano and Sennur Ulukus},
  title = {Age of Information: An Introduction and Survey},
  journal = {IEEE Journal on Selected Areas in Communications},
  volume = {39},
  number = {5},
  pages = {1183--1210},
  doi = {10.1109/JSAC.2021.3065072},
  year = {2021}
}

@article{clopperpearson,
  author = {C. J. Clopper and E. S. Pearson},
  title = {The Use of Confidence or Fiducial Limits Illustrated in the Case
           of the Binomial},
  journal = {Biometrika},
  volume = {26},
  number = {4},
  pages = {404--413},
  doi = {10.1093/biomet/26.4.404},
  year = {1934}
}

@misc{natsrequest,
  author = {{NATS Authors}},
  title = {Request-Reply},
  howpublished = {\url{https://docs.nats.io/nats-concepts/core-nats/reqreply}},
  note = {Accessed July 2026}
}

@misc{natsleaf,
  author = {{NATS Authors}},
  title = {Leaf Nodes},
  howpublished = {\url{https://docs.nats.io/running-a-nats-service/configuration/leafnodes}},
  note = {Accessed July 2026}
}

@misc{k8staints,
  author = {{The Kubernetes Authors}},
  title = {Taints and Tolerations},
  year = {2025},
  howpublished = {\url{https://v1-34.docs.kubernetes.io/docs/concepts/scheduling-eviction/taint-and-toleration/}},
  note = {Kubernetes v1.34 documentation, accessed July 2026}
}

@misc{k8spods,
  author = {{The Kubernetes Authors}},
  title = {Pod Lifecycle},
  year = {2025},
  howpublished = {\url{https://v1-34.docs.kubernetes.io/docs/concepts/workloads/pods/pod-lifecycle/}},
  note = {Kubernetes v1.34 documentation, accessed July 2026}
}

@misc{k8sjobs,
  author = {{The Kubernetes Authors}},
  title = {Jobs},
  year = {2025},
  howpublished = {\url{https://v1-34.docs.kubernetes.io/docs/concepts/workloads/controllers/job/}},
  note = {Kubernetes v1.34 documentation, accessed July 2026}
}

@misc{k8stermination,
  author = {{The Kubernetes Authors}},
  title = {Determine the Reason for Pod Failure},
  year = {2024},
  howpublished = {\url{https://v1-34.docs.kubernetes.io/docs/tasks/debug/debug-application/determine-reason-pod-failure/}},
  note = {Kubernetes v1.34 documentation, accessed July 2026}
}

@misc{k8sapi,
  author = {{The Kubernetes Authors}},
  title = {Kubernetes {API} Concepts},
  year = {2025},
  howpublished = {\url{https://v1-34.docs.kubernetes.io/docs/reference/using-api/api-concepts/}},
  note = {Kubernetes v1.34 documentation, accessed July 2026}
}

@misc{k8sobjects,
  author = {{The Kubernetes Authors}},
  title = {Object Names and {IDs}},
  year = {2025},
  howpublished = {\url{https://v1-34.docs.kubernetes.io/docs/concepts/overview/working-with-objects/names/}},
  note = {Kubernetes v1.34 documentation, accessed July 2026}
}

@misc{k8sresources,
  author = {{The Kubernetes Authors}},
  title = {Resource Management for Pods and Containers},
  year = {2026},
  howpublished = {\url{https://v1-34.docs.kubernetes.io/docs/concepts/configuration/manage-resources-containers/}},
  note = {Kubernetes v1.34 documentation, accessed July 2026}
}

@article{neymanpearson,
  author = {Jerzy Neyman and Egon S. Pearson},
  title = {On the Problem of the Most Efficient Tests of Statistical
           Hypotheses},
  journal = {Philosophical Transactions of the Royal Society of London.
             Series A},
  volume = {231},
  number = {694--706},
  pages = {289--337},
  doi = {10.1098/rsta.1933.0009},
  year = {1933}
}

@article{chen2008osa,
  author = {Yunxia Chen and Qing Zhao and Ananthram Swami},
  title = {Joint Design and Separation Principle for Opportunistic Spectrum
           Access in the Presence of Sensing Errors},
  journal = {IEEE Transactions on Information Theory},
  volume = {54},
  number = {5},
  pages = {2053--2071},
  doi = {10.1109/TIT.2008.920248},
  year = {2008}
}

@article{zhao2008periodic,
  author = {Qianchuan Zhao and Stefan Geirhofer and Lang Tong and
            Brian M. Sadler},
  title = {Opportunistic Spectrum Access via Periodic Channel Sensing},
  journal = {IEEE Transactions on Signal Processing},
  volume = {56},
  number = {2},
  pages = {785--796},
  doi = {10.1109/TSP.2007.907867},
  year = {2008}
}

@article{li2011cognitive,
  author = {Xin Li and Qianchuan Zhao and Xiaohong Guan and Lang Tong},
  title = {Optimal Cognitive Access of {Markovian} Channels under Tight
           Collision Constraints},
  journal = {IEEE Journal on Selected Areas in Communications},
  volume = {29},
  number = {4},
  pages = {746--756},
  doi = {10.1109/JSAC.2011.110407},
  year = {2011}
}

@inproceedings{bloem2015shield,
  author = {Roderick Bloem and Bettina K{\"o}nighofer and
            Robert K{\"o}nighofer and Chao Wang},
  title = {Shield Synthesis: Runtime Enforcement for Reactive Systems},
  booktitle = {Tools and Algorithms for the Construction and Analysis of
               Systems},
  series = {Lecture Notes in Computer Science},
  volume = {9035},
  pages = {533--548},
  doi = {10.1007/978-3-662-46681-0_51},
  year = {2015}
}

@inproceedings{sieve,
  author = {Xudong Sun and Wenqing Luo and Jiawei Tyler Gu and
            Aishwarya Ganesan and Ramnatthan Alagappan and Michael Gasch and
            Lalith Suresh and Tianyin Xu},
  title = {Automatic Reliability Testing for Cluster Management Controllers},
  booktitle = {16th USENIX Symposium on Operating Systems Design and
               Implementation},
  pages = {143--159},
  year = {2022},
  publisher = {USENIX Association},
  url = {https://www.usenix.org/conference/osdi22/presentation/sun}
}

@inproceedings{kubedirect,
  author = {Sheng Qi and Zhiquan Zhang and Xuanzhe Liu and Xin Jin},
  title = {{KUBEDIRECT}: Unleashing the Full Power of the Cluster Manager for
           Serverless Computing},
  booktitle = {23rd USENIX Symposium on Networked Systems Design and
               Implementation},
  pages = {1341--1357},
  year = {2026},
  publisher = {USENIX Association},
  url = {https://www.usenix.org/conference/nsdi26/presentation/qi}
}

@inproceedings{kivi,
  author = {Bingzhe Liu and Gangmuk Lim and Ryan Beckett and
            P. Brighten Godfrey},
  title = {Kivi: Verification for Cluster Management},
  booktitle = {2024 USENIX Annual Technical Conference},
  pages = {509--527},
  year = {2024},
  publisher = {USENIX Association},
  url = {https://www.usenix.org/conference/atc24/presentation/liu-bingzhe}
}

\end{document}